\documentclass[]{interact}

\usepackage{epstopdf}
\usepackage[caption=false]{subfig}

\usepackage[longnamesfirst,sort]{natbib}
\bibpunct[, ]{(}{)}{;}{a}{,}{,}
\renewcommand\bibfont{\fontsize{10}{12}\selectfont}

\theoremstyle{plain}
\newtheorem{theorem}{Theorem}[section]
\newtheorem{lemma}[theorem]{Lemma}

\theoremstyle{definition}

\theoremstyle{remark}

\newcommand{\tr}{\mbox{tr}}

\usepackage{setspace}
\usepackage{bm}
\usepackage{amsmath}
\usepackage{multirow}
\usepackage{mathrsfs}
\usepackage{amssymb}
\usepackage{float}
\usepackage{amsthm}
\usepackage{color}
\allowdisplaybreaks[4]   

\begin{document}

\articletype{ARTICLE TEMPLATE}

\begin{center}
{\Large \bf Weighted Average Ensemble for Cholesky-based Covariance Matrix Estimation} \\
Xiaoning Kang, Zhenguo Gao, Xi Liang and Xinwei Deng
\end{center}

\begin{abstract}
The modified Cholesky decomposition (MCD) is an efficient technique for estimating a covariance matrix.
However, it is known that the MCD technique often requires a pre-specified variable ordering in the estimation procedure.
In this work, we propose a weighted average ensemble covariance estimation for high-dimensional data based on the MCD technique. It can flexibly accommodate the high-dimensional case and ensure the positive definiteness property of the resultant estimate.
Our key idea is to obtain different weights for different candidate estimates by minimizing an appropriate risk function with respect to the Frobenius norm.
Different from the existing ensemble estimation based on the MCD, the proposed method  provides a sparse weighting scheme such that one can distinguish which variable orderings employed in the MCD are useful for the ensemble matrix estimate.
The asymptotically theoretical convergence rate of the proposed ensemble estimate is established under regularity conditions.
The merits of the proposed method are examined by the simulation studies and a portfolio allocation example of real stock data.
\end{abstract}


\section{Introduction}
The estimation of covariance matrix plays an important role in the multivariate statistics with a broad spectrum of applications including dimension reduction, linear discriminant analysis, social network, remote sensing, functional magnetic resonance imaging etc.
Recently, scholars have investigated the covariance estimation for a certain type of data, such as the compositional data \citep{cao2019large}, the matrix-valued data \citep{zhang2022covariance} and the spatial data \citep{kidd2022bayesian}.
However, the covariance estimation for high dimensions where the number of variables is larger than the sample size encounters challenges due to two difficulties.
One is the positive definite and symmetric properties required by the covariance matrix itself.
The other is a large amount of parameters involved in the model estimation in the sense that the number of parameters increases rapidly in the quadratic order of the model dimension.
Therefore, the covariance matrix estimation is not an easy work, and has attracted extensive interest among many scholars.

To induce the sparsity for the high-dimensional covariance matrix estimate, \cite{bickel2008covariance} suggested a hard-thresholding method which directly sets the small quantities of the sample covariance matrix to be zero, and provided its theoretical results.
In addition, with the assumption that variables are weakly correlated when they are far away in an  ordering, \cite{furrer2007estimation} and \cite{bickel2008regularized} introduced the tapering and banded covariance estimates which are consistent under some conditions.
Although the computational cost of such thresholding-based methods is small, their estimates cannot guarantee the property of the positive definiteness.
To circumvent this problem, \cite{bien2011sparse} improved the covariance estimation via imposing the $L_1$ penalty on the negative log likelihood.
Furthermore, \cite{xue2012positive} developed a covariance estimation with  additional positive definite constraint added on their objective function.
However, they are computationally intensive due to either the non-convexity of the objective function or the iterative estimation procedure.
More covariance estimation studies can be found in \cite{deng2013penalized, cai2016estimating, huang2017calibration, ledoit2020analytical, jenny2021covariance, xin2023compound}, among others.

Another powerful tool for the estimation of the high-dimensional covariance matrix is the modified Cholesky decomposition (MCD) \citep{pourahmadi1999joint} which is efficient to deal with the two challenges mentioned above.
It guarantees the positive definiteness and symmetry of the covariance estimate via the matrix decomposition, and
transforms the difficult problem of matrix estimation into the easy tasks of solving linear regressions, which is thus statistically meaningful and able to accommodate large number of variables in the high-dimensional data analysis.
In view of this, many papers in the literature have studied the covariance estimation via the MCD \citep{wu2003, huang2006, Leng2011Forward, Pedeli2015, Lv2018Smoothed, kang2021ensemble}.
However, the matrix estimates by the MCD technique depend on the variable ordering when constructing the linear regressions, implying that different variable orderings will lead to different Cholesky-based estimates.
Therefore, it requires a prior knowledge on the variable ordering of data in practice, e.g., the longitudinal data and the spatial data have a natural ordering.
This remarkably narrows down the scope of its applications, since in most cases the real data do not have a natural ordering, or we do not have any information on the variable orderings before the analysis, such as the gene data, the stock data, the industry data and the medical data.
Hence several papers have contributed to solving this ordering issue.
One idea is that a certain variable ordering can be determined based on several criteria via some data driven mechanisms before applying the MCD technique.
Typical criteria include the Isomap proposed by \cite{wagaman2009discovering},
the Bayesian Information Criterion (BIC) suggested by \cite{dellaportas2012cholesky}, the Best Permutation Algorithm (BPA) introduced by \cite{rajaratnam2013best} and so forth.
Although their methods extend the application scope of the MCD approach,
an inappropriately identified ordering will damage the performance of the subsequent Cholesky-based estimation.

In contrast to using only one single variable ordering, \cite{zheng2017cholesky} introduced another idea of employing the ensemble model via the MCD for the covariance estimation, which has been extensively studied in the later literature  \citep{Kang2020An, li2021cholesky, liang2024new, kang2024block}.
These methods have two advantages: (1) they require no prior knowledge on the variable ordering, solving the issue mentioned above;
(2) the ensemble estimators combine useful information attained from different candidate models to reduce the estimation risk caused by relying on one single variable ordering which might be incorrectly selected, thus improving the estimation accuracy.
The numerical studies have empirically shown that the Cholesky-based matrix estimates by the ensemble model perform better than the estimates using one single variable ordering.

Nonetheless, such ensemble matrix estimates by the MCD technique in general have two main disadvantages.
The first one is that they randomly generated a set of different variable orderings, and considered them to contribute equally in constructing the ensemble estimates.
That is, they assigned the same weight to the candidate estimates that are obtained by the MCD corresponding to different variable orderings.
An obvious way to improving the performance of such simply averaging estimators is to consider an ensemble model by using the optimal weights that might be different for each candidate estimate.
The second disadvantage is that they took into account all the randomly generated variable orderings for the ensemble covariance estimation, without distinguishing whether the variable orderings are useful for building the covariance estimators under the framework of the MCD.
To our best knowledge, these two disadvantages lie in all the existing Cholesky-based ensemble methods of matrix estimation.
Especially the second one might affect the estimation accuracy if some useless variable orderings are involved into the covariance estimation.
Therefore, in this paper we propose a Weighted Averaging Ensemble (WAE) for covariance estimation based on the MCD technique.
We take the advantage of the MCD framework that it is able to provide multiple candidates of the covariance estimates by considering different variable orderings.
The proposed method overcomes the first disadvantage by aggregating such candidates via the optimal weights, which are adaptively determined by minimizing a risk function in terms of the Frobenius norm.
The proposed method also solves the second disadvantage through identifying which variable orderings are the most useful for the ensemble covariance estimates by means of imposing a penalty with respect to the weights on the objective function.

The remaining of this article is organized as follows.
In Section~\ref{sec:meth}, we review the MCD for the covariance estimation, then introduce the proposed WAE methods in details, and establish the theoretical properties.
The simulations are conducted in Section~\ref{sec:sim}, and a real case study of portfolio selection is presented in Section~\ref{sec:appl}.
We conclude this work with some discussions in Section~\ref{sec:conc}.
All the technical proofs are in the Appendix.

\section{Methodology}\label{sec:meth}

\subsection{Covariance estimation via the MCD}\label{sec2.1}
Without loss of generality, we assume that $\mathbf{X}=(X_1,...,X_p)^T$ is a $p$-dimensional random vector with mean value of $\mathbf{0}$ and covariance matrix $\mathbf{\Sigma}$.
The MCD technique regresses each variable $X_j$ on its preceding errors $\varepsilon_1, \ldots, \varepsilon_{j-1}$ for $2 \leq j \leq p$.
Specifically, consider a series of linear regressions
\begin{equation}\label{eq2}
X_{j} = \sum_{i=1}^{j-1}l_{ji}\varepsilon_i + \varepsilon_{j},~~~j=2,...,p,
\end{equation}
where $\bm l_{j} = (l_{j1}, \ldots, l_{j(j-1)})^T$ is the vector consisted of regression coefficients, and $\varepsilon_{j}$ is the error term for the $j$th regression with its expectation of 0 and variance of $d_{j}^{2}$.
In addition, let $\varepsilon_1 = X_1$ and $d_{1}^{2} = Var(X_1)$. Define by a diagonal matrix  $\mathbf{D}=diag(d_1^2,...,d_p^2)$ the variance matrix of the error vector $\boldsymbol{\varepsilon} = (\varepsilon_1,...,\varepsilon_p)^T$. Let $\mathbf{L}=(l_{ji})_{p\times p}$ be a lower triangular matrix whose diagonal elements are all ones, where $l_{jj}=1$ and $l_{ji}=0$ if $i>j$.
As a result, the linear regressions \eqref{eq2} can be written in the matrix form of $\mathbf{X} = \mathbf{L}\boldsymbol{\varepsilon}$.
The covariance matrix can be thus written as  $\mathbf{\Sigma} = Var(\mathbf{X}) = Var(\mathbf{L}\boldsymbol{\varepsilon}) = \mathbf{L}\mathbf{D}\mathbf{L}^T$.
By this means, the MCD technique reduces the difficulty of covariance estimation by decomposing it into the estimation of the Cholesky factor matrices $\mathbf{D}$ and $\mathbf{L}$, which can be obtained easily by solving the $p-1$ linear regressions in \eqref{eq2}.

Let $\mathbf{x_1,...,x_n}$ be $n$ independent observations from the random vector $\mathbf{X}$, and $\mathbb{X} = (\mathbf{x_1,...,x_n})^T$ be the $n \times p$ data matrix.
Define the $j$th column of $\mathbb{X}$ by $\mathbf{x}^{(j)}$.
Denote by $\mathbf{e}^{(j)}$ the residuals of the $j$th regression for $j \geq 2$, and $\mathbf{e}^{(1)} = \mathbf{x}^{(1)}$.
Let $\mathbb{Z}^{(j)} = (\mathbf{e}^{(1)}, \ldots, \mathbf{e}^{(j-1)})$ be the matrix containing the first $j-1$ residuals.
For the high-dimensional data, the Lasso regularization is often used to shrink the estimated regression coefficients \citep{huang2006, chang2010estimation}
\begin{equation}\label{eq5}
\hat{\bm l}_{j} = \arg \min_{ \bm l_{j} } \| \mathbf{x}^{(j)} -  \mathbb{Z}^{(j)} \bm l_{j} \|_{2}^{2}
+ \lambda_{j} \| \bm l_{j} \|_{1},~~~j = 2, \ldots, p,
\end{equation}
where $\lambda_{j} \geq 0$ is a tuning parameter.
The symbols $\| \cdot \|_{1}$ and $\| \cdot \|_{2}$ stand for the vector $L_1$ and $L_2$ norms.
$\mathbf{e}^{(j)} = \mathbf{x}^{(j)} -  \mathbb{Z}^{(j)} \hat{\bm l}_{j}$ is used to construct the residuals for the last column of $\mathbb{Z}^{(j+1)}$.
The estimate $\hat{d}_j^2$ is calculated as the sample variance of   $\mathbf{e}^{(j)}$.
Consequently, the covariance estimate is $\mathbf{\hat{\Sigma}} = \mathbf{\hat{L}} \mathbf{\hat{D}} \mathbf{\hat{L}}^T$, where $\mathbf{\hat{L}}$ is constructed with $\hat{\bm l}_{j}$ as its $j$th row, and $\mathbf{\hat{D}} = diag(\hat{d}_1^2, ..., \hat{d}_p^2)$.

\subsection{The proposed estimators}\label{sec2.2}

Since the Cholesky factor matrix estimates $\hat{\mathbf{L}}$ and $\hat{\mathbf{D}}$ depend on the variable ordering $X_1 ,..., X_p$ as one can see from Equation \eqref{eq2}, we apply the ensemble model to solve this issue by considering multiple orderings for the MCD technique.
Define a mapping $\Pi:\{1,2,...,p\} \to \{1,2,...,p\}$ as $\{1,2,...,p\} \to \{\Pi(1),...,\Pi(p)\}$, and let $\mathbf{P}_{\Pi}$ be the corresponding permutation matrix.
Let $(\mathbf{L}_{\Pi}, \mathbf{D}_{\Pi})$ be the matrices of $(\mathbf{L}, \mathbf{D})$ obtained by the MCD using the variable ordering permutation $\Pi$.
Accordingly we have a covariance estimate for $\mathbf{\Sigma}$ under $\Pi$ by transforming $\mathbf{\hat{\Sigma}}_{\Pi} = \mathbf{\hat{L}}_{\Pi} \mathbf{\hat{D}}_{\Pi} \mathbf{\hat{L}}_{\Pi}^T$ back to the original ordering in the following way
\begin{equation}\label{eq3}
	\begin{aligned}
		\mathbf{\hat{\Sigma}}
		= \mathbf{P}_{\Pi} \mathbf{\hat{\Sigma}}_{\Pi} \mathbf{P}_{\Pi}^T
		= \mathbf{P}_{\Pi} \mathbf{\hat{L}}_{\Pi} \mathbf{\hat{D}}_{\Pi} \mathbf{\hat{L}}_{\Pi}^T \mathbf{P}_{\Pi}^T
		= (\mathbf{P}_{\Pi} \mathbf{\hat{L}}_{\Pi} \mathbf{P}_{\Pi}^T) (\mathbf{P}_{\Pi} \mathbf{\hat{D}}_{\Pi} \mathbf{P}_{\Pi}^T) (\mathbf{P}_{\Pi} \mathbf{\hat{L}}_{\Pi}^T \mathbf{P}_{\Pi}^T)  = \mathbf{\hat{L}} \mathbf{\hat{D}} \mathbf{\hat{L}}^T,
	\end{aligned}
\end{equation}
where $\mathbf{\hat{L}} = \mathbf{P}_{\Pi} \mathbf{\hat{L}}_{\Pi} \mathbf{P}_{\Pi}^T$ and $\mathbf{\hat{D}} = \mathbf{P}_{\Pi} \mathbf{\hat{D}}_{\Pi} \mathbf{P}_{\Pi}^T$.
Suppose we consider $M$ different variable orderings $\mathbf{\Pi}_k, k = 1,...,M$, and denote the corresponding estimates $\mathbf{\hat{\Sigma}}, \mathbf{\hat{L}}$, and $\mathbf{\hat{D}}$ in \eqref{eq3} by $\mathbf{\tilde{\Sigma}}_k, \mathbf{\tilde{L}}_k$, and $\mathbf{\tilde{D}}_k$.
\cite{zheng2017cholesky} proposed a model averaging estimate $\mathbf{\hat{\Sigma}}_{zheng} = \frac{1}{M} \sum_{k=1}^{M} \mathbf{\tilde{\Sigma}}_k$, which has a relatively smaller variance due to the averaging and displays better performance than the Cholesky-based covariance estimate with one single variable ordering.
However, since the Cholesky-based covariance estimates $\mathbf{\tilde{\Sigma}}_k$'s with different orderings may have different estimation biases, and \cite{zheng2017cholesky} simply assigns equal weights to each candidate estimate $\mathbf{\tilde{\Sigma}}_k$, their method hence ignores that these estimates under different variable orderings may have different effects to the ensemble model.
Therefore, we propose the following WAE estimate
\begin{equation}\label{eq:prop}
\mathbf{\hat{\Sigma}}_{wae} = \sum_{k=1}^{M} \omega_k \mathbf{\tilde{\Sigma}}_k,
\end{equation}
where $\boldsymbol{\omega} = (\omega_1, ..., \omega_M)^T$ is the vector of weights, which plays a critical role in producing a good and reliable ensemble covariance estimate.
To obtain the optimal weights, we consider to minimize the difference between the ensemble estimator and the true covariance matrix under the Frobenius norm.
It is worth noting that such a strategy is also used in shrinking the estimators of covariance matrix \citep{ledoit2003improved}.
Specifically, we employ the following quadratic loss function as the risk function to determine the optimal weights by
\begin{align*}
\hat{\boldsymbol{\omega}} & = \arg \min_{ \boldsymbol{\omega} } R(\boldsymbol{\omega}) = \arg \min_{ \boldsymbol{\omega} }  E||\omega_1 \mathbf{\tilde{\Sigma}}_1 + ... + \omega_M \mathbf{\tilde{\Sigma}}_M - \mathbf{\Sigma}||_F^2  \\
&\quad \mbox{s.t.} \quad \sum_{i=1}^{M} \omega_i = 1 \quad \mbox{and} \quad \omega_i \geq 0,
\end{align*}
where $E(\cdot)$ represents the expectation.
Let $\mathbf{\Sigma} = (\sigma_{ij})_{p \times p}$ and $\mathbf{\tilde{\Sigma}}_k = (\tilde{\sigma}_{ij}^k)_{p \times p}$. We decompose the risk function as
\begin{align*}
R(\boldsymbol{\omega}) &= E||\omega_1 \mathbf{\tilde{\Sigma}}_1 + ... + \omega_M \mathbf{\tilde{\Sigma}}_M - \mathbf{\Sigma}||_F^2 \\
&= \sum_{i=1}^p \sum_{j=1}^p E (\omega_1 \tilde{\sigma}_{ij}^1 + ... + \omega_M \tilde{\sigma}_{ij}^M - \sigma_{ij})^2 \\
&= \sum_{i=1}^p \sum_{j=1}^p \{ Var(\omega_1 \tilde{\sigma}_{ij}^1 + ... + \omega_M \tilde{\sigma}_{ij}^M) + [E(\omega_1 \tilde{\sigma}_{ij}^1 + ... + \omega_M \tilde{\sigma}_{ij}^M - \sigma_{ij})]^2 \}  \\
&= \sum_{i=1}^p \sum_{j=1}^p \{ \omega_1^2 Var(\tilde{\sigma}_{ij}^1) + ... + \omega_M^2 Var(\tilde{\sigma}_{ij}^M) + \sum_{h \neq \ell} \omega_h \omega_{\ell} Cov (\tilde{\sigma}_{ij}^h, \tilde{\sigma}_{ij}^{\ell}) \\
&~~~+ [\omega_1 E(\tilde{\sigma}_{ij}^1) + ... + \omega_M E(\tilde{\sigma}_{ij}^M) - \sigma_{ij}]^2 \} \\
&= \sum_{i=1}^{p} \sum_{j=1}^{p} \sum_{k=1}^{M} \omega_k^2 [Var(\tilde{\sigma}_{ij}^k) + E^2(\tilde{\sigma}_{ij}^k)] + \sum_{i=1}^{p}\sum_{j=1}^{p}\sum_{h \neq \ell} \omega_h \omega_\ell [Cov(\tilde{\sigma}_{ij}^h,\tilde{\sigma}_{ij}^\ell) + E(\tilde{\sigma}_{ij}^h)E(\tilde{\sigma}_{ij}^\ell)] \\
&\quad-\sum_{i=1}^{p}\sum_{j=1}^{p} \sum_{k=1}^{M}2\omega_k E(\tilde{\sigma}_{ij}^k)\sigma_{ij} + \sum_{i=1}^{p}\sum_{j=1}^{p}\sigma_{ij}^2.
\end{align*}
To write the risk function $R(\boldsymbol{\omega})$ in a quadratic form of a matrix, we define a symmetric matrix $\mathbf{A}= (a_{mn})$ and a vector $\mathbf{b} = (b_{k})$, where
\begin{align*}
a_{mn} = \sum_{i=1}^{p} \sum_{j=1}^{p} Cov(\tilde{\sigma}_{ij}^m, \tilde{\sigma}_{ij}^n) + E(\tilde{\sigma}_{ij}^m) E(\tilde{\sigma}_{ij}^n)~~~\mbox{and}~~~
b_{k} = -2\sum_{i=1}^{p} \sum_{j=1}^{p} E(\tilde{\sigma}_{ij}^k)\sigma_{ij}.
\end{align*}
Consequently the risk function $R(\boldsymbol{\omega})$ is rewritten as
\begin{equation*}	
R(\boldsymbol{\omega}) = \boldsymbol{\omega}^T \mathbf{A} \boldsymbol{\omega} + \mathbf{b}^T \boldsymbol{\omega} + \sum_{i=1}^{p}\sum_{j=1}^{p}\sigma_{ij}^2.
\end{equation*}
The optimal weight vector $\boldsymbol{\omega}$ is estimated as the minimizer of  $R(\boldsymbol{\omega})$, that is,
\begin{equation}\label{obj_weight}
\hat{\boldsymbol{\omega}} = \arg \min_{ \boldsymbol{\omega} } \boldsymbol{\omega}^T \mathbf{A} \boldsymbol{\omega} + \mathbf{b}^T \boldsymbol{\omega} \quad \mbox{s.t.} \quad \sum_{i=1}^{M} \omega_i = 1 ~~~\mbox{and}~~~ \omega_i \geq 0.
\end{equation}
This optimization problem can be easily solved using the function $solve.QP(\cdot)$ in the package $quadprog$ of R software.

%

As we mention in Section 1, not all the variable orderings would be useful in the weighted ensemble model for estimating the covariance matrix based on the MCD.
Some sets of linear regressions of \eqref{eq2} from certain variable orderings may play no roles in the sense that using them in the MCD would increase the model complexity but not improve the estimation accuracy.
It is equivalent to saying that some candidate estimates $\mathbf{\tilde{\Sigma}}_k$ might contribute little or nothing to the weighted ensemble covariance estimation in Equation \eqref{eq:prop}.
In such case, we hence need to rule out such variable orderings that could be useless for the data analysis, and simultaneously to simplify the model complexity. It can be realized by setting the corresponding weights to zeros.
To this end, we adopt the Lasso \citep{tibshirani1996regression} and the adaptive Lasso ideas \citep{zou2006adaptive} and suggest to impose two types of penalties on Equation \eqref{obj_weight}, which yields
\begin{align}\label{obj_weight_lasso}
\hat{\boldsymbol{\omega}} = \arg \min_{ \boldsymbol{\omega} } \boldsymbol{\omega}^T \mathbf{A} \boldsymbol{\omega} + \mathbf{b}^T \boldsymbol{\omega} + \phi \sum_{i=1}^{M} |\omega_i| \quad \mbox{s.t.} \quad \sum_{i=1}^{M} \omega_i = 1 ~~~\mbox{and}~~~ \omega_i \geq 0
\end{align}
and
\begin{align}\label{obj_weight_penalty}
\hat{\boldsymbol{\omega}}^{\ast} = \arg \min_{ \boldsymbol{\omega} } \boldsymbol{\omega}^T \mathbf{A} \boldsymbol{\omega} + \mathbf{b}^T \boldsymbol{\omega} + \xi \sum_{i=1}^{M} \theta_i |\omega_i| \quad \mbox{s.t.} \quad \sum_{i=1}^{M} \omega_i = 1 ~~~\mbox{and}~~~ \omega_i \geq 0,
\end{align}
where $\phi \geq 0$ and $\xi \geq 0$ are tuning parameters.
A little algebra on \eqref{obj_weight_lasso} shows
\begin{align*}
\hat{\boldsymbol{\omega}} &= \arg \min_{ \boldsymbol{\omega} } \boldsymbol{\omega}^T \mathbf{A} \boldsymbol{\omega} + \mathbf{b}^T \boldsymbol{\omega} + \phi \sum_{i=1}^{M} \omega_i \\
&= \arg \min_{ \boldsymbol{\omega} } \boldsymbol{\omega}^T \mathbf{A} \boldsymbol{\omega} + \mathbf{b}^T \boldsymbol{\omega} + \phi \quad \mbox{s.t.} \quad \sum_{i=1}^{M} \omega_i = 1 ~~~\mbox{and}~~~ \omega_i \geq 0,
\end{align*}
which is equivalent to the optimization problem \eqref{obj_weight}.
It is thereby very interesting to see that the proposed WAE method with weight $\hat{\boldsymbol{\omega}}$ solved from \eqref{obj_weight} overcomes two disadvantages of existing methods that are mentioned in the Introduction by simultaneously completing the tasks of allocating different weights to the estimates $\mathbf{\tilde{\Sigma}}_k$ and picking up the most useful variable orderings for the ensemble estimation.

For the optimization problem \eqref{obj_weight_penalty}, let $\bm \theta = (\theta_1, ..., \theta_M)^T$, and then it follows that
\begin{align*}
\hat{\boldsymbol{\omega}}^{\ast} &= \arg \min_{ \boldsymbol{\omega} } \boldsymbol{\omega}^T \mathbf{A} \boldsymbol{\omega} + \mathbf{b}^T \boldsymbol{\omega} + \xi \sum_{i=1}^{M} \theta_i \omega_i \\
&= \arg \min_{ \boldsymbol{\omega} } \boldsymbol{\omega}^T \mathbf{A} \boldsymbol{\omega} + \mathbf{b}^T \boldsymbol{\omega} + \xi \boldsymbol{\theta}^T \boldsymbol{\omega}  \\
&= \arg \min_{ \boldsymbol{\omega} } \boldsymbol{\omega}^T \mathbf{A} \boldsymbol{\omega} + (\mathbf{b} + \xi \bm \theta)^T \boldsymbol{\omega} \quad \mbox{s.t.} \quad \sum_{i=1}^{M} \omega_i = 1 ~~~\mbox{and}~~~ \omega_i \geq 0.
\end{align*}
Similar to the adaptive Lasso idea, we suggest the value of $\bm \theta = 1 / \hat{\boldsymbol{\omega}}$, where $\hat{\boldsymbol{\omega}}$ is the solution from the optimization problem \eqref{obj_weight}.
As a result, we induce a sparsity in the estimated weight vector, which automatically distinguishes and thus selects a set of the most important candidate estimates $\mathbf{\tilde{\Sigma}}_k$, or equivalently the most useful variable orderings, for the ensemble covariance estimation under the framework of MCD based on the data themselves.

The optimal value of the tuning parameter $\xi$ in the optimization problem \eqref{obj_weight_penalty} can be determined by minimizing the negative likelihood function of covariance matrix $\bm \Sigma$.
Specifically, consider a set of candidate values for $\xi$, denoted as $\mathcal{A}_\xi = \{\xi_1, \ldots, \xi_N \}$.
Let $\hat{\bm \Sigma}_{wae^\ast}(\xi_t)$ represent the proposed weighted ensemble covariance estimate computed from Equation \eqref{eq:prop} with its weights solved from the optimization \eqref{obj_weight_penalty} using tuning parameter $\xi_t, t = 1, \ldots, N$.
We then obtain $Q_t = \log|\hat{\bm \Sigma}_{wae^\ast}(\xi_t)| + \mbox{trace} [\hat{\bm \Sigma}^{-1}_{wae^\ast}(\xi_t) \bm S]$, where $\bm S$ is the sample covariance matrix.
Subsequently denote $\tilde{t} = \arg \min_t \{Q_t, t = 1, \ldots, N\}$.
Consequently, the optimal value of $\xi$ is $\xi_{\tilde{t}}$, and the resultant covariance estimate is $\hat{\bm \Sigma}_{wae^\ast}(\xi_{\tilde{t}})$.
That is, the proposed covariance estimate is
\begin{equation*}
\mathbf{\hat{\Sigma}}_{wae^\ast} = \sum_{k=1}^{M} \tilde{\omega}_k^\ast \mathbf{\tilde{\Sigma}}_k,
\end{equation*}
where $\tilde{\boldsymbol{\omega}}^\ast = (\tilde{\omega}_1^\ast, \ldots, \tilde{\omega}_M^\ast)^T$ is the solution of Equation \eqref{obj_weight_penalty} using tuning parameter $\xi_{\tilde{t}}$.

\subsection{Theoretical properties}\label{sec:theorem}
This section establishes the convergence properties of the proposed WAE estimators.
To facilitate the expression of the theoretical results, some assumptions are made on the true model.
Let $\mathbf{\Sigma}_0 = \mathbf{L}_0 \mathbf{D}_0 \mathbf{L}^T_0$ represent the underlying covariance matrix and its MCD.
Similarly, let $\bm \Sigma_{0\Pi_{k}} = \mathbf{L}_{0\Pi_{k}} \mathbf{D}_{0\Pi_{k}} \mathbf{L}_{0\Pi_{k}}^T$ be the MCD of the underlying covariance matrix regarding a variable order $\Pi_{k}$.
Define $\mathscr{C}_{\Pi_k} = \{(i,j): j<i, l_{0ij}^{\Pi_k} \neq 0\}$ to be a set indexing the nonzero elements in the lower triangular part of the matrix $\mathbf{L}_{0 \Pi_k} = (l_{0ij}^{\Pi_k})$.
Then denote the maximum cardinality of $\mathscr{C}_{\Pi_k}$ by $s$.
Additionally, denote the singular values of matrix $\mathbf{\Sigma}_0$ by $sv_{1}(\mathbf{\Sigma}_0) \geq \ldots \geq sv_{p}(\mathbf{\Sigma}_0)$ in a decreasing order.
In order to establish the theoretical property, we assume the regularity conditions as listed below
\begin{itemize}
\item[C1]: The singular values of $\mathbf{\Sigma}_0$ are bounded. That is, there exist constants $l_{1}$ and $u_{1}$ such that $0 < l_{1} < sv_{p}(\mathbf{\Sigma}_0) \leq sv_{1}(\mathbf{\Sigma}_0) < u_{1} < \infty$.
\item[C2]: The tuning parameters $\lambda_{j}$'s in \eqref{eq5} satisfy $\sum_{j=1}^p \lambda_{j} = O(\sqrt{\log(p) / n})$.
\item[C3]: $(s + p) \log (p) = o(n)$.
\end{itemize}
The condition C1 is commonly used to guarantee the positive definiteness of $\mathbf{\Sigma}_0$ in the literature.
The conditions C2 and C3 are used to derive the consistency properties of the proposed estimators.
Let $a_n \asymp b_n$ represent that two sequences $a_n$ and $b_n$ are the same order.
Now we present the main results.

\begin{theorem}\label{theory1}
Suppose that the data are independently and identically distributed from $N(\bm 0, \mathbf{\Sigma}_0)$. Under the conditions C1 - C3, we have
\begin{align*}
\| \mathbf{\hat{\Sigma}}_{wae} - \mathbf{\Sigma}_0 \|_{F} \asymp \| \mathbf{\hat{\Sigma}}_{wae^\ast} - \mathbf{\Sigma}_0 \|_{F} = O_{p} \left( \sqrt{\frac{(s + p) \log (p)}{n}} \right).
\end{align*}
\end{theorem}

Theorem \ref{theory1} demonstrates the asymptotic convergence rates of the proposed WAE estimators with respect to the Frobenius norm.
Such rate is parallel with that of some existing works \citep{rothman2008sparse, lam2009sparsistency}.
The review paper of \cite{cai2016estimating} pointed out that the high-dimensional covariance estimators have the convergence rate in the form of $O_{p}(\sqrt{\frac{(\kappa + p) \log (p)}{n}})$ in terms of the Frobenius norm, where $\kappa$ is a measure of the sparsity for the underlying covariance matrix.

\subsection{Estimates of matrix $\mathbf{A}$ and vector $\mathbf{b}$}\label{sec2.3}
In order to solve the weights in the optimizations \eqref{obj_weight} and \eqref{obj_weight_penalty}, we need to find consistent estimators for $Var(\tilde{\sigma}_{ij}^k)$, $Cov(\tilde{\sigma}_{ij}^k, \tilde{\sigma}_{ij}^\ell)$, $E(\tilde{\sigma}_{ij}^k)E(\tilde{\sigma}_{ij}^\ell)$ and $E(\tilde{\sigma}_{ij}^k)\sigma_{ij}$ in the matrix $\mathbf{A}$ and vector $\mathbf{b}$.
Define $\varphi^k_{ij}=AsyVar(\sqrt{n}\tilde{\sigma}_{ij}^k)$, $\rho^{k\ell}_{ij} = AsyCov(\sqrt{n}\tilde{\sigma}_{ij}^k,\sqrt{n}\tilde{\sigma}_{ij}^\ell)$,  $\gamma^{k\ell}_{ij}=E(\tilde{\sigma}_{ij}^k)E(\tilde{\sigma}_{ij}^\ell)$ and  $\eta^k_{ij}=E(\tilde{\sigma}_{ij}^k)\sigma_{ij}$, where $AsyVar$ and $AsyCov$ represent the asymptotic variance and covariance.
Let $m_j = \sum_{i=1}^n \mathbf{x}_i^{(j)} / n$ be the mean value of the $j$th column of data matrix $\mathbb{X}$.
Standard asymptotic theory provides consistent estimators for $\varphi^k_{ij}$, $\rho^{k\ell}_{ij}$, $\gamma^{k\ell}_{ij}$ and $\eta^{k}_{ij}$, as shown in the following lemmas.

\begin{lemma}\label{lemma1}
A consistent estimator for $\varphi^k_{ij}$ is given by
\begin{equation*} \hat{\varphi}^k_{ij}=\frac{1}{n}\sum_{t=1}^{n}[(x_{ti}-m_i)(x_{tj}-m_j)-\tilde{\sigma}_{ij}^k]^2.
\end{equation*}
\end{lemma}

\begin{proof}
Without loss of generality, we assume $i<j$.
Because $X_i = \sum_{k=1}^{i}l_{ik} \varepsilon_k$ and $X_j = \sum_{k=1}^{j}l_{jk} \varepsilon_k$, we have $Cov(X_i,X_j) = \sum_{k=1}^{i}l_{ik}l_{jk}d_k^2$.
Hence $s_{ij} = \sum_{k=1}^{i}\hat{l}_{ik} \hat{l}_{jk} \hat{d}_k^2 = \tilde{\sigma}_{ij}^k$ without regularization \citep{zheng2017cholesky}.
By the Lemma 1 in \cite{ledoit2003improved}, $\hat{\varphi}^k_{ij}$ converges in probability to $\varphi^k_{ij}$.
\end{proof}

\begin{lemma}\label{lemma2}
A consistent estimator for $\rho^{k\ell}_{ij}$ is given by
\begin{equation*}
\hat{\rho}^{k\ell}_{ij}=\frac{1}{n}\sum_{t=1}^{n}[(x_{ti}-m_i)(x_{tj}-m_j)-\tilde{\sigma}_{ij}^k][(x_{ti}-m_i)(x_{tj}-m_j)-\tilde{\sigma}_{ij}^\ell].
\end{equation*}
\end{lemma}

\begin{proof}
It is an easy extension of the result from Lemma \ref{lemma1}.
\end{proof}


Under conditions C1 - C3 and from Theorem \ref{theory1}, estimators for $\gamma^{k\ell}_{ij}$ and $\eta^k_{ij}$ are $\tilde{\sigma}_{ij}^k \tilde{\sigma}_{ij}^\ell$ and $\tilde{\sigma}_{ij}^k s_{ij}$, where $s_{ij}$ is the $(i,j)$th element of the sample covariance matrix.
Consequently, we can replace $\varphi_{ij}^k$, $\rho_{ij}^{k\ell}$, $\gamma_{ij}^{k\ell}$ and $\eta_{ij}^k$ with their corresponding estimators into the optimizations \eqref{obj_weight} and \eqref{obj_weight_penalty} to solve the weights for each $\mathbf{\tilde{\Sigma}}_k$.

\section{Simulation}\label{sec:sim}
In this section, the performances of the proposed WAE estimators $\mathbf{\hat{\Sigma}}_{wae}$ and $\mathbf{\hat{\Sigma}}_{wae^\ast}$ are evaluated by the comparison with the method developed by \cite{zheng2017cholesky}, denoted as $\mathbf{\hat{\Sigma}}_{zheng}$, which will demonstrate the advantages of the optimal weights solved from the optimization problems \eqref{obj_weight} and \eqref{obj_weight_penalty} over the equal weights for the average ensemble idea in the MCD.
Let $\mathbb{I}_{\{\cdot\}}$ represent the indicator function.
We consider seven different setups for the true covariance matrix in the simulation studies.

\begin{itemize}
	\item[] {Scenario-1 (Identity Structure)}: $\mathbf{\Sigma_1} = \mathbf{I}$ is the identity matrix.
	\item[] {Scenario-2 (Compact Banded Structure)}:
$\mathbf{\Sigma_2} = \mathbf{B}^T\mathbf{B}$, where $\mathbf{B}=(b_{st})$ with $b_{st} = \mathbb{I}_{\{s=t\}} + 0.8\mathbb{I}_{\{s-t=1\}} + 0.6\mathbb{I}_{\{s-t=2\}}$, is the second-order moving average structure.
	\item[] {Scenario-3 (Permuted Banded Structure)}: $\mathbf{\Sigma_3}$ is obtained by randomly permutating rows and corresponding columns of $\mathbf{\Sigma_2}$.
	\item[] {Scenario-4 (Loose Banded Structure)}: $\mathbf{\Sigma_4}=\mathbf{B}^T\mathbf{B}$, where $\mathbf{B}=(b_{st})$ with $b_{st}=\mathbb{I}_{\{s=t\}}+0.8\mathbb{I}_{\{s-t=1\}}+0.6\mathbb{I}_{\{2 \leq s-t \leq 5\}}$.
	\item[] {Scenario-5 (Block Diagonal Structure)}: $\mathbf{\Sigma_5}=(\sigma_{st})$, $\sigma_{st}=\mathbb{I}_{\{s=t\}}+0.5\mathbb{I}_{\{s \neq t, s \leq 20, t \leq 20\}}$.
	\item[] {Scenario-6 (Dense Structure)}: $\mathbf{\Sigma_6}=\mathbf{B}\mathbf{B}^T$, $\mathbf{B}=(b_{st})$ is a unit lower triangular matrix with $b_{st}$ independently generated from normal distribution $N(0, 0.2)$.
    \item[] {Scenario-7 (Compound Structure)}: $\mathbf{\Sigma_7}=(\sigma_{st})$,
$\sigma_{st}=\mathbb{I}_{\{s=t\}} + 0.5\mathbb{I}_{\{s \neq t\}}$.
\end{itemize}

Scenario-1 is the identity matrix, indicating variables are independent from each other.
Scenarios-2 and 4 are banded matrices, indicating each variable is only correlated with its most nearby variables.
Scenario-5 has a compound structure with element value 0.5 on the upper left corner.
Scenario-6 represents a randomly generated matrix.
Scenario-7 indicates that each variable is correlated with others.
Therefore, the true covariance matrices in this study include very sparse cases (e.g. Scenario-1, Scenario-2, Scenario-3), moderate sparse cases (e.g. Scenario-4, Scenario-5) as well as the dense cases (e,g, Scenario-6, Scenario-7).
In order to measure the accuracy of covariance estimates, we consider the following three loss functions
\begin{align*}
\mbox{F} = || \mathbf{\hat{\Sigma}} - \mathbf{\Sigma} ||_F, ~~~
\mbox{MAE} = \frac{1}{p^2} \sum_{i,j} | \hat{\sigma}_{ij} - \sigma_{ij} | ~~~\mbox{and}~~~
L_2 = \lambda_{\max}(\mathbf{\hat{\Sigma}}-\mathbf{\Sigma}),
\end{align*}
where $\mathbf{\hat{\Sigma}} = (\hat{\sigma}_{ij})$ stands for a covariance estimate $\mathbf{\Sigma} = (\sigma_{ij})$, and $\lambda_{\max}$ represents the maximum eigenvalue of the matrix.
For each scenario, the data are independently generated from multivariate normal distribution $N(\bm 0, \bm \Sigma)$ with the sample size $n=50$, and the dimensionality $p=50, 100$.
The number of variable orderings $M$ is set to be 30 as suggested by \cite{zheng2017cholesky}.
Such 30 variable orderings are uniformly sampled, that is, they are randomly generated with equal probability from all $p!$ possible orderings.
The LASSO technique of \eqref{eq5} in the proposed method is implemented by the function $glmnet(\cdot)$ in the R software with the value of its tuning parameter selected by the function itself.
The set $\mathcal{A}_\xi$ for the tuning parameter $\xi$ in the optimization problem \eqref{obj_weight_penalty} contains an arithmetic sequence starting from 0.01 to 3 with the common difference 0.05.
Table~\ref{tab0} shows the averaged values and corresponding standard errors of loss functions for each compared estimator over 50 replications.

\begin{table}[h]
    \centering
	\caption{Means and standard errors (in parentheses) of loss functions from each estimate for normal data.}
    \resizebox{\textwidth}{!}{ 
    \begin{tabular}{crrrrrrrrr}
    \hline
		\multirow{2}*{Scenario} & \multicolumn{3}{c}{$(n,p) = (50,50)$} & \multicolumn{3}{c}{$(n,p) = (50,100)$} \\
        ~ & $\mathbf{\hat{\Sigma}}_{zheng}$ & $\mathbf{\hat{\Sigma}}_{wae}$ & $\mathbf{\hat{\Sigma}}_{wae^\ast}$ & $\mathbf{\hat{\Sigma}}_{zheng}$ & $\mathbf{\hat{\Sigma}}_{wae}$ & $\mathbf{\hat{\Sigma}}_{wae^\ast}$\\
        \hline
        \quad \quad \quad \quad F \\
        \hline
        $\mathbf{\Sigma_1}$ & 4.36(0.03) & 3.43(0.09) & 3.59(0.08) & 4.88(0.02) & 4.36(0.06) & 4.44(0.06) \\
        $\mathbf{\Sigma_2}$ & 10.2(0.05) & 7.80(0.06) & 7.90(0.07) & 15.3(0.05) & 12.2(0.05) & 12.3(0.06) \\
        $\mathbf{\Sigma_3}$ & 10.3(0.05) & 7.94(0.06) & 8.00(0.07) & 15.3(0.05) & 12.2(0.06) & 12.3(0.07) \\
        $\mathbf{\Sigma_4}$ & 27.9(0.08) & 16.6(0.20) & 16.8(0.21) & 41.4(0.08) & 28.1(0.18)  & 28.6(0.18) \\
        $\mathbf{\Sigma_5}$ & 8.81(0.04) & 6.06(0.17) & 6.64(0.16) & 9.27(0.03) & 7.22(0.14) & 7.68(0.13) \\
		$\mathbf{\Sigma_6}$ & 20.6(0.06) & 16.8(0.09) & 17.1(0.09) & 51.5(0.08) &
        48.6(0.09) & 49.1(0.09) \\
        $\mathbf{\Sigma_7}$ & 23.0(0.05) & 15.6(0.44) & 16.2(0.42) & 47.9(0.05) & 37.3(0.61) & 37.3(0.61) \\
		\hline
        \quad \quad \quad \quad $L_2$ \\
        \hline
        $\mathbf{\Sigma_1}$ & 4.00(0.03) & 2.72(0.11) & 2.86(0.11) & 4.15(0.02) & 3.24(0.07) & 3.30(0.07) \\
        $\mathbf{\Sigma_2}$ & 3.64(0.02) & 3.09(0.03) & 3.08(0.04) & 3.84(0.02) & 3.43(0.02) & 3.45(0.02) \\
        $\mathbf{\Sigma_3}$ & 3.68(0.02) & 3.15(0.04) & 3.17(0.04) & 3.84(0.02) & 3.44(0.03) & 3.45(0.03) \\
        $\mathbf{\Sigma_4}$ & 13.1(0.05) & 9.05(0.16) & 9.10(0.16) & 13.8(0.03) & 10.8(0.09)  & 10.9(0.09) \\
        $\mathbf{\Sigma_5}$ & 8.64(0.04) & 5.64(0.20) & 6.19(0.17) & 8.90(0.03) & 6.55(0.16) & 6.90(0.15) \\
		$\mathbf{\Sigma_6}$ & 9.49(0.05) & 7.46(0.09) & 7.56(0.09) & 18.0(0.07) &
        16.8(0.11) & 16.9(0.10) \\
        $\mathbf{\Sigma_7}$ & 22.9(0.05) & 15.3(0.45) & 16.0(0.44) & 47.8(0.05) & 37.1(0.62) & 37.0(0.62) \\
        \hline
        \quad \quad \quad \quad MAE \\
        \hline
        $\mathbf{\Sigma_1}$ & 0.026(0.000) & 0.024(0.001) & 0.024(0.001) & 0.012(0.000) & 0.011(0.000) & 0.011(0.000) \\
        $\mathbf{\Sigma_2}$ & 0.083(0.000) & 0.082(0.001) & 0.081(0.001) & 0.050(0.000) & 0.048(0.000) & 0.048(0.000) \\
        $\mathbf{\Sigma_3}$ & 0.084(0.000) & 0.083(0.001) & 0.083(0.001) & 0.050(0.000) & 0.048(0.000) & 0.048(0.000) \\
        $\mathbf{\Sigma_4}$ & 0.271(0.001) & 0.203(0.002) & 0.203(0.002) & 0.149(0.000) & 0.128(0.001) & 0.128(0.001) \\
        $\mathbf{\Sigma_5}$ & 0.079(0.000) & 0.056(0.002) & 0.061(0.001) & 0.024(0.000) & 0.021(0.000) & 0.022(0.000) \\
		$\mathbf{\Sigma_6}$ & 0.315(0.001) & 0.266(0.001) & 0.269(0.001) &
        0.398(0.001) & 0.381(0.001) & 0.385(0.001) \\
        $\mathbf{\Sigma_7}$ & 0.458(0.001) & 0.299(0.009) & 0.311(0.009) & 0.478(0.001) & 0.366(0.007) & 0.364(0.007) \\
        \hline
    \end{tabular}}
    \label{tab0}
\end{table}

\begin{table}[h]
    \centering
	\caption{Means and standard errors (in parentheses) for the number of non-zeros weights for the proposed estimates.}
    \begin{tabular}{crrrrrrrrr}
    \hline
		\multirow{2}*{Scenario} & \multicolumn{2}{c}{$(n,p) = (50,50)$} & \multicolumn{2}{c}{$(n,p) = (50,100)$} \\
        ~  & $\mathbf{\hat{\Sigma}}_{wae}$ & $\mathbf{\hat{\Sigma}}_{wae^\ast}$ & $\mathbf{\hat{\Sigma}}_{wae}$ & $\mathbf{\hat{\Sigma}}_{wae^\ast}$\\
        \hline
        $\mathbf{\Sigma_1}$ & 7.12(0.27) & 3.68(0.21) & 8.46(0.33) & 5.22(0.17) \\
        $\mathbf{\Sigma_2}$ & 9.02(0.30) & 4.86(0.22) & 9.86(0.25) & 5.34(0.15) \\
        $\mathbf{\Sigma_3}$ & 8.94(0.26) & 6.00(0.15) & 9.76(0.28) & 6.18(0.15) \\
        $\mathbf{\Sigma_4}$ & 7.64(0.27) & 4.66(0.16) & 8.62(0.25) & 4.48(0.14) \\
        $\mathbf{\Sigma_5}$ & 5.52(0.28) & 2.32(0.08) & 6.32(0.26) & 2.86(0.16) \\
		$\mathbf{\Sigma_6}$ & 12.1(0.27) & 6.76(0.21) & 10.5(0.32) & 6.70(0.15) \\
        $\mathbf{\Sigma_7}$ & 5.58(0.29) & 3.24(0.12) & 4.62(0.21) & 2.92(0.16) \\
		\hline
    \end{tabular}
    \label{tab1}
\end{table}

It can be seen that overall two proposed WAE methods are substantially superior to the estimate $\mathbf{\hat{\Sigma}}_{zheng}$ in terms of $F$ and $L_2$ loss functions for all the considered covariance structures, and they are at least comparable with or slightly better than $\mathbf{\hat{\Sigma}}_{zheng}$ with respect to the MAE loss.
By comparing the results of Scenario-2 and Scenario-3, we observe that the proposed WAE methods have the similar performances when we permute the rows and columns of the banded covariance matrix randomly.
Besides, by the comparison of two proposed methods, we observe that the estimate $\mathbf{\hat{\Sigma}}_{wae^\ast}$ is almost comparable with but sometimes slightly worse the estimate $\mathbf{\hat{\Sigma}}_{wae}$.
A possible reason can be found in Table \ref{tab1}, which displays the number of non-zero elements in the estimated weights $\hat{\boldsymbol{\omega}}$ and $\hat{\boldsymbol{\omega}}^\ast$ from the optimizations \eqref{obj_weight} and \eqref{obj_weight_penalty}, respectively.
We observe that the ensemble estimate $\mathbf{\hat{\Sigma}}_{wae^\ast}$ aggregates less number of candidate estimates than $\mathbf{\hat{\Sigma}}_{wae}$, implying that it reduces the model complexity and thus picks up the most useful variable orderings for constructing the covariance estimate.
Therefore, the estimate $\mathbf{\hat{\Sigma}}_{wae}$ might be a little more accurate sometimes, and the estimate $\mathbf{\hat{\Sigma}}_{wae^\ast}$ leads to a simpler model.
Furthermore from Table \ref{tab1} we conclude that in practice it is no need to calculate every candidate estimate $\mathbf{\tilde{\Sigma}}_k$ with respect to all the $p!$ variable orderings, which is very computationally expensive, since there are only a few orderings that are the most important for constructing the ensemble matrix estimate.

\begin{table}[h]
    \centering
	\caption{Means and standard errors (in parentheses) of loss functions from each estimate for non-normal data.}
    \resizebox{\textwidth}{!}{ 
    \begin{tabular}{ccrrrrrrrr}
    \hline
		& \multirow{2}*{Case} & \multicolumn{3}{c}{$(n,p) = (50,50)$} & \multicolumn{3}{c}{$(n,p) = (50,100)$} \\
        ~ && $\mathbf{\hat{\Sigma}}_{zheng}$ & $\mathbf{\hat{\Sigma}}_{wae}$ & $\mathbf{\hat{\Sigma}}_{wae^\ast}$ & $\mathbf{\hat{\Sigma}}_{zheng}$ & $\mathbf{\hat{\Sigma}}_{wae}$ & $\mathbf{\hat{\Sigma}}_{wae^\ast}$\\
        \hline
        F&(I) & 16.4(1.32) & 14.0(1.47) & 14.4(1.53) & 46.3(0.14) & 37.7(2.30) & 38.7(2.52) \\
        &(II) & 23.3(0.13) & 17.0(0.62) & 17.0(0.63) & 48.0(0.07) & 37.9(0.71) & 37.8(0.71) \\
        \hline
        $L_2$&(I) & 16.1(1.30) & 12.6(1.44) & 12.9(1.50) & 46.0(0.14) & 34.6(2.35) & 35.1(2.60) \\
        &(II) & 23.1(0.06) & 16.5(0.56) & 16.5(0.57) & 47.9(0.05) & 37.7(0.71) & 37.6(0.71) \\
        \hline
        MAE&(I) & 0.323(0.026) & 0.233(0.023) & 0.238(0.024) & 0.460(0.001) & 0.328(0.013) & 0.334(0.014) \\
        &(II) & 0.464(0.002) & 0.327(0.013) & 0.326(0.013) & 0.479(0.001) & 0.372(0.007) & 0.371(0.008) \\
        \hline
    \end{tabular}}
    \label{tab:robust}
\end{table}

Notice that the MCD technique empirically does not require the normal property of data to estimate the covariance matrix.
It is thus interesting to examine the robust performance of the proposed WAE methods for the non-normal data.
Now we conduct a simulation study using $\bm \Sigma_7$ as the true covariance matrix, and consider the following two cases of data generating processes with $n = 50$ and $p = 50,100$.

$\textbf{Case}$ (I). Independently generate data from $t$ distribution with degrees of freedom 4, as well as mean $\bm 0$ and $\bm \Sigma_7$ as its scale matrix.

$\textbf{Case}$ (II). Independently generate data from a mixed distribution $(1 - b) N(\bm 0, \bm \Sigma_7) + b \bm T$, where $b$ is a random number following Bernoulli distribution with the probability of success equal to 0.1. Here $\bm T$ represents $t$ distribution with degrees of freedom 4, as well as mean $\bm 5$ and the identity matrix as its scale matrix.

The data in Case (I) come from a heavy-tailed distribution, and Case (II) presents a set of contaminated data.
Table \ref{tab:robust} displays the averaged values and corresponding standard errors of loss functions from each competitor over 50 replications.
Overall, the proposed WAE methods outperform the estimate $\mathbf{\hat{\Sigma}}_{zheng}$ substantially under the considered losses, implying the advantages of the ensemble idea with the optimally selected weights for different candidate covariance estimates.

To sum up, the simulation studies demonstrate that the proposed WAE methods improve the performance of the ensemble estimation for covariance matrix under the framework of MCD technique by assigning each covariance candidate estimate with different weights that are obtained from minimizing the risk functions as in \eqref{obj_weight} and \eqref{obj_weight_penalty}.

\section{Case Study of Portfolio Allocation}\label{sec:appl}
To further explore the performances of the proposed WAE methods, we apply them to analyze a real stock market data, where the estimated covariance matrix is subsequently used for the decision on the portfolio allocation.
	
The common portfolio strategy \citep{Markowitz1952} tries to minimize the risk of an expected return level by diversifying the investment of various assets.
Let $\mathbf{w} = (w_1, ..., w_p)^T$ represent the proportions of a set of assets in the portfolio, and $\mathbf{\Sigma}$ denote the volatility matrix of the returns in the asset pool.
The minimum variance portfolio selection problem is formulated as
\begin{equation}\label{eq4}
\min \mathbf{w}^T \mathbf{\Sigma} \mathbf{w} ~~~
\mbox{s.t.}\quad \sum_{i=1}^p w_i = 1 ~~~\mbox{and}~~~ w_i \geq 0,~i=1, ..., p.
\end{equation}
Since the objective function in \eqref{eq4} involves the covariance matrix, an ill-conditioned matrix estimate may lead to an unstable solution of $\mathbf{w}$, and thus greatly magnifies the error of portfolio allocation.
In addition, when a large amount of stocks is considered for the investment, a valid estimate of volatility matrix that accommodates for the high dimensions is needed.
Furthermore, because there is no natural variable ordering among stocks, the idea of ensemble model by the MCD is suitable for this analysis.
We hence apply the proposed WAE methods for the covariance estimation, which may lead to a good subsequent portfolio allocation strategy.
We compare the performances of the proposed methods with $\mathbf{\hat{\Sigma}}_{zheng}$ as in the simulation studies.

The data are consisted of the stock returns collected weekly from years 2020 to 2021 from the CSI300 Index, which contain $p = 100$ randomly selected stocks with $n = 104$ observations (52 weeks for each year).
The first $j$ observations are used to construct the covariance estimate  $\hat{\mathbf{\Sigma}}_j$, $j = 52, ..., 103$, for each method.
Then the estimated portfolio weight $\hat{\mathbf{w}}_{j+1}$ is the solution of the optimization problem \eqref{eq4} by replacing $\mathbf{\Sigma}$ with its estimate $\hat{\mathbf{\Sigma}}_j$.
We measure the performance of the estimated portfolio by the realized averaged weekly return (AWR) for the year of 2021, which is defined as
\begin{align*}
\mbox{AWR} = \frac{1}{52}\sum_{j = 52}^{103} \hat{\mathbf{w}}_{j+1}^T \mathbf{x}_{j+1},
\end{align*}
as well as their standard errors (SE) and the information ratio AWR/SE.
Although a larger AWR is the main pursue for the investors since it is the profit, a higher information ratio is also desired as it reflects both profit and risk.

\begin{table}[h]
	\centering
	\caption{The performances of weekly returns of portfolio for compared methods.}
	\vspace{0.05in}
	\begin{tabular}{c c c c c}
		\hline
		& AWR & SE & AWR/SE &    \\
		\hline
        $\mathbf{\hat{\Sigma}}_{zheng}$       & 0.647 & 0.565 & 1.145   \\
		$\mathbf{\hat{\Sigma}}_{wae}$         & 0.764 & 0.604 & 1.265  \\
		$\mathbf{\hat{\Sigma}}_{wae^\ast}$    & 0.795 & 0.606 & 1.296  \\
		\hline
	\end{tabular}
	\label{tab:portfolio}
\end{table}

\begin{figure}[h]
	\centering
	\includegraphics[width=1\linewidth]{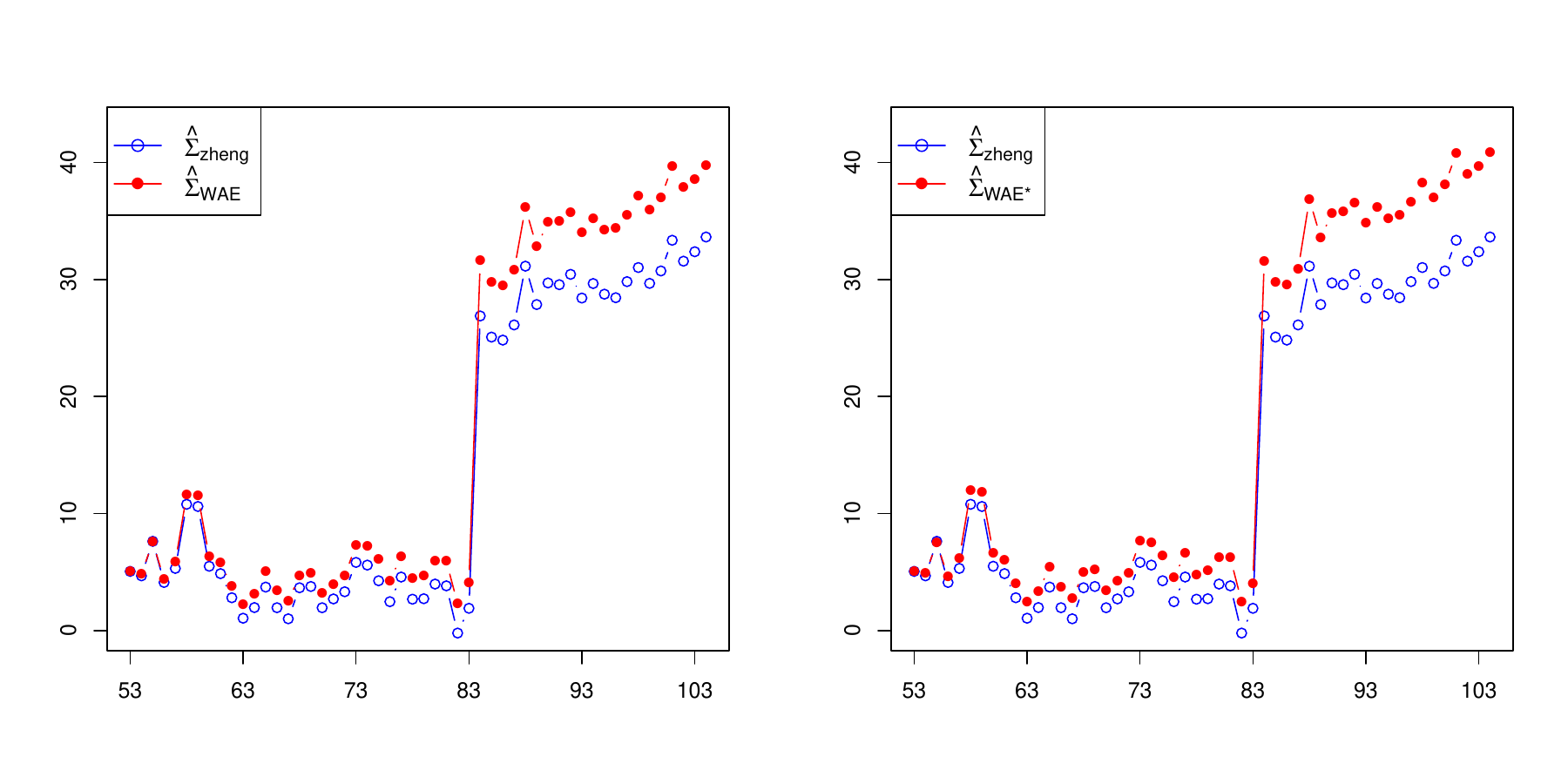}
	\caption{The cumulative weekly returns from compared methods for portfolio study.}
	\label{cu_prof}
\end{figure}

Table~\ref{tab:portfolio} presents the compared portfolio performances obtained by the proposed WAE methods and $\mathbf{\hat{\Sigma}}_{zheng}$ for year 2021.
We observe that although the proposed models produce a slightly larger SE, they have  higher AWR and AWR/SD values compared with the method of $\mathbf{\hat{\Sigma}}_{zheng}$, which indicates that the proposed methods gain more realized returns.
Furthermore, Figure~\ref{cu_prof} are the plots of the cumulative weekly returns for different compared methods.
Because the most parts of lines produced by the estimates $\mathbf{\hat{\Sigma}}_{wae^\ast}$ and $\mathbf{\hat{\Sigma}}_{wae}$ overlap each other as they behave similarly in this study, we plot their cumulative weekly returns separately in two plots in order to see them clearly in comparison with the estimate $\mathbf{\hat{\Sigma}}_{zheng}$.
From the plots it is observed that the cumulative weekly returns of the proposed WAE methods are consistently higher than that of the estimate $\mathbf{\hat{\Sigma}}_{zheng}$.
Especially as the investment time increases, the advantages of the proposed methods over the estimate $\mathbf{\hat{\Sigma}}_{zheng}$ are much more evidenced.
Therefore the proposed WAE methods are able to provide a valid strategy for the portfolio allocation for this set of stock data, verifying the merits of using different weights in the ensemble model of MCD for this study.

\section{Discussion}\label{sec:conc}
In this work, we propose a weighted average ensemble covariance matrix estimation via the MCD technique. The weights are adaptively determined by minimizing a risk function in terms of the Frobenius norm.
The risk function can be transformed into a matrix quadratic form with respect to the weights such that it can be solved easily by the off-the-shelf softwares.
At the same time, the proposed methods are able to automatically distinguish which variable orderings utilized in the MCD are useful, and thus preclude from the ensemble model the candidate covariance estimates corresponding to the useless variable orderings to reduce the model complexity.
The asymptotic convergence rates of the proposed methods are established, and
the numerical studies demonstrate the merits of the proposed methods that use the optimal weights under the MCD framework.

There are several research directions along this line in the future work.
Firstly, we would like to point out that although a simulation is performed to examine the robustness of the proposed methods, they are not designed for the non-normal data.
One may employ some robust techniques such as Huber method instead of LASSO in Equation \eqref{eq5} to decrease the impact of the outliers or heavy-tailed data on the linear models, hence leading to a robust covariance estimate \citep{wang2023novel}.
Secondly, we might improve the proposed WAE by constructing an appropriate objective function regarding the Cholesky factor matrices $\mathbf{L}$ and $\mathbf{D}$ rather than $\bm \Sigma$ for the weights optimization.
Such weights are used to compute the weighted averaging values on $\mathbf{\tilde{L}}_k$ and $\mathbf{\tilde{D}}_k$, i.e., $\mathbf{\hat{L}} = \sum_{k=1}^{M} \omega_k \mathbf{\tilde{L}}_k$ and $\mathbf{\hat{D}} = \sum_{k=1}^{M} \omega_k \mathbf{\tilde{D}}_k$.
Then a covariance estimate is $\mathbf{\hat{\Sigma}} = \mathbf{\hat{L}} \mathbf{\hat{D}} \mathbf{\hat{L}}^T$ which would have a smaller variance than $\mathbf{\hat{\Sigma}}_{wae}$, hence leading to a more accurate estimate \citep{Kang2020An}.
However, solving such objective functions in terms of $\mathbf{L}$ and $\mathbf{D}$ is more complicated.
Another idea to determine a covariance matrix estimate from matrices $\mathbf{\tilde{\Sigma}}_k$ is based on the concepts of Bures-Wasserstein distance and the Fr\'{e}chet mean, which are used to measure the central tendency of a set of positive semi-definite matrices \citep{zheng2023barycenter}.
However, whether such technique is suitable for the Cholesky-based covariance estimation still needs further investigation.

%

\section*{Data Availability Statement}
The data that support the findings of this study are available from the corresponding author upon reasonable request.

\section*{Disclosure statement}
There is no conflict of interest in this work.

%

\bibliographystyle{ECA_jasa}
\bibliography{Ref_WAE}

@article{jenny2021covariance,
  title={Covariance estimation via fiducial inference},
  author={Jenny Shi, W and Hannig, Jan and Lai, Randy CS and Lee, Thomas CM},
  journal={Statistical Theory and Related Fields},
  volume={5},
  number={4},
  pages={316--331},
  year={2021},
  publisher={Taylor \& Francis}
}

@article{zheng2023barycenter,
  title={Barycenter Estimation of Positive Semi-Definite Matrices with Bures-Wasserstein Distance},
  author={Zheng, Jingyi and Huang, Huajun and Yi, Yuyan and Li, Yuexin and Lin, Shu-Chin},
  journal={arXiv preprint arXiv:2302.14618},
  year={2023}
}

@article{tibshirani1996regression,
  title={Regression shrinkage and selection via the lasso},
  author={Tibshirani, Robert},
  journal={Journal of the Royal Statistical Society Series B: Statistical Methodology},
  volume={58},
  number={1},
  pages={267--288},
  year={1996},
  publisher={Oxford University Press}
}

@article{wagaman2009discovering,
  title={Discovering sparse covariance structures with the isomap},
  author={Wagaman, AS and Levina, E},
  journal={Journal of Computational and Graphical Statistics},
  volume={18},
  number={3},
  pages={551--572},
  year={2009},
  publisher={Taylor \& Francis}
}

@article{zou2006adaptive,
  title={The adaptive lasso and its oracle properties},
  author={Zou, Hui},
  journal={Journal of the American statistical association},
  volume={101},
  number={476},
  pages={1418--1429},
  year={2006},
  publisher={Taylor \& Francis}
}

@article{liang2024new,
  title={A new approach for ultrahigh dimensional precision matrix estimation},
  author={Liang, Wanfeng and Zhang, Yuhao and Wang, Jiyang and Wu, Yue and Ma, Xiaoyan},
  journal={Journal of Statistical Planning and Inference},
  pages={106164},
  year={2024},
  publisher={Elsevier}
}

@article{kang2024block,
  title={On block Cholesky decomposition for sparse inverse covariance estimation},
  author={Kang, Xiaoning and Lian, Jiayi and Deng, Xinwei},
  journal={Statistica Sinica},
  pages={1--19},
  year={2025}
}

@article{wang2023novel,
  title={A novel robust estimation for high-dimensional precision matrices},
  author={Wang, Shaoxin and Xie, Chaoping and Kang, Xiaoning},
  journal={Statistics in Medicine},
  volume={42},
  number={5},
  pages={656--675},
  year={2023},
  publisher={Wiley Online Library}
}

@article{li2021cholesky,
  title={A Cholesky-based sparse covariance estimation with an application to genes data},
  author={Li, Chunshi and Yang, Mo and Wang, Mingqiu and Kang, Hong and Kang, Xiaoning},
  journal={Journal of Biopharmaceutical Statistics},
  volume={31},
  number={5},
  pages={603--616},
  year={2021},
  publisher={Taylor \& Francis}
}

@article{cao2019large,
  title={Large covariance estimation for compositional data via composition-adjusted thresholding},
  author={Cao, Yuanpei and Lin, Wei and Li, Hongzhe},
  journal={Journal of the American Statistical Association},
  volume={114},
  number={526},
  pages={759--772},
  year={2019},
  publisher={Taylor \& Francis}
}

@article{lam2009sparsistency,
  title={Sparsistency and rates of convergence in large covariance matrix estimation},
  author={Lam, Clifford and Fan, Jianqing},
  journal={Annals of statistics},
  volume={37},
  number={6B},
  pages={4254--4278},
  year={2009},
  publisher={NIH Public Access}
}

@article{huang2017calibration,
  title={A calibration method for non-positive definite covariance matrix in multivariate data analysis},
  author={Huang, Chao and Farewell, Daniel and Pan, Jianxin},
  journal={Journal of Multivariate Analysis},
  volume={157},
  pages={45--52},
  year={2017},
  publisher={Elsevier}
}

@article{zhang2022covariance,
  title={Covariance estimation for matrix-valued data},
  author={Zhang, Yichi and Shen, Weining and Kong, Dehan},
  journal={Journal of the American Statistical Association},
  volume={117},
  pages={1--12},
  year={2022},
  publisher={Taylor \& Francis}
}

@article{kidd2022bayesian,
  title={Bayesian Nonstationary and Nonparametric Covariance Estimation for Large Spatial Data (with Discussion)},
  author={Kidd, Brian and Katzfuss, Matthias},
  journal={Bayesian Analysis},
  volume={17},
  number={1},
  pages={291--351},
  year={2022},
  publisher={International Society for Bayesian Analysis}
}

@article{cai2016estimating,
  title={Estimating structured high-dimensional covariance and precision matrices: Optimal rates and adaptive estimation},
  author={Cai, T Tony and Ren, Zhao and Zhou, Harrison H},
  journal={Electronic Journal of Statistics},
  volume={10},
  number={1},
  pages={1--59},
  year={2016},
  publisher={Institute of Mathematical Statistics and Bernoulli Society}
}

@article{bickel2008covariance,
  title={Covariance regularization by thresholding},
  author={Bickel, Peter J and Levina, Elizaveta},
  journal={The Annals of Statistics},
  volume={36},
  number={6},
  pages={2577--2604},
  year={2008},
  publisher={Institute of Mathematical Statistics}
}

@article{bickel2008regularized,
  title={Regularized estimation of large covariance matrices},
  author={Bickel, Peter J and Levina, Elizaveta},
  journal={The Annals of Statistics},
  volume={36},
  number={1},
  pages={199--227},
  year={2008},
  publisher={Institute of Mathematical Statistics}
}

@article{furrer2007estimation,
  title={Estimation of high-dimensional prior and posterior covariance matrices in Kalman filter variants},
  author={Furrer, Reinhard and Bengtsson, Thomas},
  journal={Journal of Multivariate Analysis},
  volume={98},
  number={2},
  pages={227--255},
  year={2007},
  publisher={Elsevier}
}

@article{xue2012positive,
  title={Positive-definite $L_1$-penalized estimation of large covariance matrices},
  author={Xue, Lingzhou and Ma, Shiqian and Zou, Hui},
  journal={Journal of the American Statistical Association},
  volume={107},
  number={500},
  pages={1480--1491},
  year={2012},
  publisher={Taylor \& Francis}
}

@article{pourahmadi1999joint,
  title={Joint mean-covariance models with applications to longitudinal data: Unconstrained parameterisation},
  author={Pourahmadi, Mohsen},
  journal={Biometrika},
  volume={86},
  number={3},
  pages={677--690},
  year={1999},
  publisher={Oxford University Press}
}

@article{dellaportas2012cholesky,
  title={Cholesky-GARCH models with applications to finance},
  author={Dellaportas, Petros and Pourahmadi, Mohsen},
  journal={Statistics and Computing},
  volume={22},
  number={4},
  pages={849--855},
  year={2012},
  publisher={Springer}
}

@article{rajaratnam2013best,
  title={Best permutation analysis},
  author={Rajaratnam, Bala and Salzman, Julia},
  journal={Journal of Multivariate Analysis},
  volume={121},
  pages={193--223},
  year={2013},
  publisher={Elsevier}
}

@article{chang2010estimation,
  title={Estimation of covariance matrix via the sparse Cholesky factor with lasso},
  author={Chang, Changgee and Tsay, Ruey S},
  journal={Journal of Statistical Planning and Inference},
  volume={140},
  number={12},
  pages={3858--3873},
  year={2010},
  publisher={Elsevier}
}

@article{zheng2017cholesky,
  title={Cholesky-based model averaging for covariance matrix estimation},
  author={Zheng, Hao and Tsui, Kam-Wah and Kang, Xiaoning and Deng, Xinwei},
  journal={Statistical Theory and Related Fields},
  volume={1},
  number={1},
  pages={48--58},
  year={2017},
  publisher={Taylor \& Francis}
}

@article{kang2021ensemble,
  title={Ensemble sparse estimation of covariance structure for exploring genetic disease data},
  author={Kang, Xiaoning and Wang, Mingqiu},
  journal={Computational Statistics \& Data Analysis},
  volume={159},
  pages={107220},
  year={2021},
  publisher={Elsevier}
}

@article{ledoit2003improved,
  title={Improved estimation of the covariance matrix of stock returns with an application to portfolio selection},
  author={Ledoit, Olivier and Wolf, Michael},
  journal={Journal of empirical finance},
  volume={10},
  number={5},
  pages={603--621},
  year={2003},
  publisher={Elsevier}
}

@article{rothman2008sparse,
  title={Sparse permutation invariant covariance estimation},
  author={Rothman, Adam J and Bickel, Peter J and Levina, Elizaveta and Zhu, Ji},
  journal={Electronic Journal of Statistics},
  volume={2},
  pages={494--515},
  year={2008},
  publisher={Institute of Mathematical Statistics and Bernoulli Society}
}

@article{bien2011sparse,
  title={Sparse estimation of a covariance matrix},
  author={Bien, Jacob and Tibshirani, Robert J},
  journal={Biometrika},
  volume={98},
  number={4},
  pages={807--820},
  year={2011},
  publisher={Oxford University Press}
}

@article{xin2023compound,
  title={A compound decision approach to covariance matrix estimation},
  author={Xin, Huiqin and Zhao, Sihai Dave},
  journal={Biometrics},
  volume={79},
  number={2},
  pages={1201--1212},
  year={2023},
  publisher={Wiley Online Library}
}

@article{ledoit2020analytical,
  title={Analytical nonlinear shrinkage of large-dimensional covariance matrices},
  author={Ledoit, Olivier and Wolf, Michael},
  journal={The Annals of Statistics},
  volume={48},
  number={5},
  pages={3043--3065},
  year={2020},
  publisher={Institute of Mathematical Statistics}
}

@article{deng2013penalized,
  title={Penalized covariance matrix estimation using a matrix-logarithm transformation},
  author={Deng, Xinwei and Tsui, Kam-Wah},
  journal={Journal of Computational and Graphical Statistics},
  volume={22},
  number={2},
  pages={494--512},
  year={2013},
  publisher={Taylor \& Francis}
}

@article{wu2003,
  title={Nonparametric estimation of large covariance matrices of longitudinal data},
  author={Weibiao Wu and Mohsen Pourahmadi},
  journal={Biometrika},
  volume={90},
  number={4},
  pages={831--844},
  year={2003},
  publisher={Biometrika Trust},
  mrnumber={2024760}
}

@article{huang2006,
  title={Covariance matrix selection and estimation via penalised normal likelihood},
  author={Huang, J.Z. and Liu, N. and Pourahmadi, M. and Liu, L.},
  journal={Biometrika},
  volume={93},
  number={1},
  pages={85--98},
  year={2006},
  publisher={Biometrika},
}

@article{Leng2011Forward,
  title={Forward adaptive banding for estimating large covariance matrices},
  author={Leng, C. and Li, B.},
  journal={Biometrika},
  volume={98},
  number={4},
  pages={821-830},
  year={2011},
}

@article{Pedeli2015,
  title={Two Cholesky-log-GARCH models for multivariate volatilities},
  author={Pedeli, X.and Fokianos, K. and Pourahmadi, M.},
  journal={Statistical Modeling},
  volume={15},
  pages={233--255},
  year={2015},
}

@article{Lv2018Smoothed,
  title={Smoothed empirical likelihood inference via the modified Cholesky decomposition for quantile varying coefficient models with longitudinal data},
  author={Lv, Jing and Guo, Chaohui and Wu, Jibo},
  journal={TEST},
  pages={1-34},
  year={2018},
}

@article{kang2021variable,
  title={On variable ordination of Cholesky-based estimation for a sparse covariance matrix},
  author={Kang, Xiaoning and Deng, Xinwei},
  journal={Canadian Journal of Statistics},
  volume={49},
  number={2},
  pages={283--310},
  year={2021},
  publisher={Wiley Online Library}
}

@article{Kang2020An,
  title={An improved modified cholesky decomposition approach for precision matrix estimation},
  author={Kang, X. and Deng, X.},
  journal={Journal of Statistical Computation and Simulation},
  volume={90},
  number={3},
  pages={443-464},
  year={2020},
}

@article{Markowitz1952,
author = {Markowitz, Harry},
title = {PORTFOLIO SELECTION*},
journal = {The Journal of Finance},
volume = {7},
number = {1},
pages = {77-91},
year = {1952}
}

\section*{Appendix}
In this section, we provide the technical proofs of the main theoretical results in the paper.
To prove Theorem \ref{theory1}, we need to present three lemmas.

\begin{lemma}\label{lemma1_app}
Assume a positive definite matrix $\bm \Sigma$ has its modified Cholesky decomposition $\bm \Sigma = \bm L \bm D \bm L^T$. Under condition C1 (that is, the singular values of $\bm \Sigma$ are bounded), there exist constants $l_2$ and $u_2$ such that
\begin{align*}
&l_2 < sv_{p}(\bm L) \leq sv_{1}(\bm L) < u_2 \\
&l_2 < sv_{p}(\bm D) \leq sv_{1}(\bm D) < u_2.
\end{align*}
\end{lemma}

\noindent The proof of Lemma \ref{lemma1_app} is similar to that of Lemma 3 in \cite{kang2024block}, thus is omitted here.

\begin{lemma}\label{lemma2_app}
Suppose that the data are independently and identically distributed from $N(\bm 0, \mathbf{\Sigma}_0)$. Under the conditions C1 and C2, the estimates of Cholesky factor matrices $\hat{\bm L}_{\Pi}$ and $\hat{\bm D}_{\Pi}$ corresponding to a permutation $\Pi$ have the following consistent properties
\begin{align*}
& \| \hat{\bm L}_{\Pi} - \bm L_{0\Pi} \|_{F} = O_{p}(\sqrt{s \log (p) / n}),  \\
& \| \hat{\bm D}_{\Pi} - \bm D_{0\Pi} \|_{F} = O_{p}(\sqrt{p \log (p) / n}).
\end{align*}
\end{lemma}

\begin{proof}
To simplify the notations, we prove it without the symbol $\Pi$ since the conclusions hold for any variable orderings.
From Equation \eqref{eq2}, it is easy to see
$\bm \epsilon = \bm L^{-1} \bm X \sim N(\bm 0, \bm D)$.
Hence, the penalized negative log likelihood, up to some constant, is
\begin{align*}
\Phi(\bm D, \bm L) = (\log \left| \bm D \right| + \tr [(\bm L^T)^{-1} \bm D^{-1} \bm L^{-1} \bm S] + \sum_{j=1}^{p} \lambda_{j} \sum_{k<j}|l_{jk}|.
\end{align*}
Define $G(\Delta_{L}, \Delta_{D}) = \Phi(\bm D_{0} + \Delta_{D}, \bm L_{0} + \Delta_{L}) - \Phi(\bm D_{0}, \bm L_{0})$.
Let $\mathcal{B}_{1} = \{ \Delta_{L}: \| \Delta_{L} \|_{F}^{2} \leq \kappa_{1} s \log(p) / n \}$ and
$\mathcal{B}_{2} = \{ \Delta_{D}: \| \Delta_{D} \|_{F}^{2} \leq \kappa_{2} p \log(p) / n \}$, where $\kappa_{1} \geq 0$ and $\kappa_{2} \geq 0$ are constants.
Denote $\partial \mathcal{B}_{1}$ and $\partial \mathcal{B}_{2}$ as the boundaries of $\mathcal{B}_{1}$ and $\mathcal{B}_{2}$.
We will show that for each $\Delta_{L} \in \partial \mathcal{B}_{1}$ and $\Delta_{D} \in \partial \mathcal{B}_{2}$, probability $\Pr(G(\Delta_{L}, \Delta_{D}) > 0)\rightarrow $ 1 as $n \rightarrow \infty$ for sufficiently large $\kappa_{1}$ and $\kappa_{2}$.
Additionally, since $G(\Delta_{L}, \Delta_{D}) = 0$ when $\Delta_{L} = 0$ and $\Delta_{D} = 0$, the $G(\Delta_{L}, \Delta_{D})$ achieves its minimum value when
$\Delta_{L} \in \mathcal{B}_{1}$ and $\Delta_{D} \in \mathcal{B}_{2}$.
That is $\| \Delta_{L} \|_{F}^{2} = O_{p} (s \log(p) / n)$ and $\| \Delta_{D} \|_{F}^{2} = O_{p} (p \log(p) / n)$.

Assume $\| \Delta_{L} \|_{F}^{2} = \kappa_{1} s \log(p) / n$ and $\| \Delta_{D} \|_{F}^{2} = \kappa_{2} p \log(p) / n$.
According to Lemma \ref{lemma1_app} and condition C1,
there exists a constant $\delta$ satisfying $0 < 1/\delta < sv_{p}(\bm L_{0}) \leq sv_{1}(\bm L_{0}) < \delta < \infty$ and $0 < 1/\delta < sv_{p}(\bm D_{0}) \leq sv_{1}(\bm D_{0}) < \delta < \infty$.
Let $\bm D = \bm D_{0} + \Delta_{D}$ and $\bm L = \bm L_{0} + \Delta_{L}$,
then we decompose $G(\Delta_{L}, \Delta_{D})$ into five parts and then bound them separately.
\begin{align*}
G(\Delta_{L}, \Delta_{D}) &= \Phi(\bm D, \bm L) - \Phi(\bm D_{0}, \bm L_{0}) \\
&= \log \left| \bm D \right| - \log \left| \bm D_{0} \right| + \tr ((\bm L^T)^{-1} \bm D^{-1} \bm L^{-1} \bm S) - \tr ((\bm L_{0}^T)^{-1} \bm D^{-1}_{0} \bm L^{-1}_{0} \bm S) \\
&~~~+ \sum_{j=1}^{p} \lambda_{j} \sum_{k<j}|l_{jk}| - \sum_{j=1}^{p} \lambda_{j} \sum_{k<j}|l_{0jk}| \\
&= \log \left| \bm D \right| - \log \left| \bm D_{0} \right| + \tr[(\bm D^{-1} - \bm D_{0}^{-1})\bm D_{0}] - \tr[(\bm D^{-1} - \bm D_{0}^{-1})\bm D_{0}] \\
&~~~+ \tr [(\bm L^T)^{-1} \bm D^{-1} \bm L^{-1} \bm S] - \tr [(\bm L_{0}^T)^{-1} \bm D^{-1}_{0} \bm L^{-1}_{0} \bm S] \\
&~~~+ \sum_{j=1}^{p} \lambda_{j} \sum_{k<j}|l_{jk}| - \sum_{j=1}^{p} \lambda_{j} \sum_{k<j}|l_{0jk}| \\
&= \log \left| \bm D \right| - \log \left| \bm D_{0} \right| + \tr[(\bm D^{-1} - \bm D_{0}^{-1})\bm D_{0}] + \tr [(\bm L^T)^{-1} \bm D^{-1} \bm L^{-1} \bm S] \\
&~~~ - \tr [(\bm L^T)^{-1} \bm D^{-1}_{0} \bm L^{-1} \bm S] + \tr [(\bm L^T)^{-1} \bm D^{-1}_{0} \bm L^{-1} \bm S] \\
&~~~+ \sum_{j=1}^{p} \lambda_{j} \sum_{k<j}|l_{jk}| - \sum_{j=1}^{p} \lambda_{j} \sum_{k<j}|l_{0jk}|  \\
&= \log \left| \bm D \right| - \log \left| \bm D_{0} \right| + \tr[(\bm D^{-1} - \bm D_{0}^{-1})\bm D_{0}] + \tr(\bm D^{-1} - \bm D_{0}^{-1})[\bm L^{-1} (\bm S - \bm \Sigma_{0}) (\bm L^T)^{-1}] \\
&~~~ + \tr [\bm D_{0}^{-1} (\bm L^{-1} (\bm S - \bm \Sigma_{0}) (\bm L^T)^{-1} - \bm L^{-1}_{0} (\bm S - \bm \Sigma_{0}) (\bm L_{0}^T)^{-1})] \\
&~~~ + \tr [\bm D_{0}^{-1} (\bm L^{-1} \bm \Sigma_{0} (\bm L^T)^{-1} - \bm L^{-1}_{0} \bm \Sigma_{0} (\bm L_{0}^T)^{-1})] + \tr(\bm D^{-1} - \bm D_{0}^{-1})(\bm L^{-1} \bm \Sigma_{0} (\bm L^T)^{-1} - \bm D_{0}) \\
&~~~ + \sum_{j=1}^{p} \lambda_{j} \sum_{k<j}|l_{jk}| - \sum_{j=1}^{p} \lambda_{j} \sum_{k<j}|l_{0jk}|  \\
&= \log \left| \bm D \right| - \log \left| \bm D_{0} \right| + \tr[(\bm D^{-1} - \bm D_{0}^{-1})\bm D_{0}] + \tr(\bm D^{-1} - \bm D_{0}^{-1})[\bm L^{-1} (\bm S - \bm \Sigma_{0}) (\bm L_{0}^T)^{-1}] \\
&~~~ + \tr [\bm D_{0}^{-1} (\bm L^{-1} (\bm S - \bm \Sigma_{0}) (\bm L^T)^{-1} - \bm L^{-1}_{0} (\bm S - \bm \Sigma_{0}) (\bm L_{0}^T)^{-1})] \\
&~~~ + \tr [\bm D^{-1} (\bm L^{-1} \bm \Sigma_{0} (\bm L^T)^{-1} - \bm L^{-1}_{0} \bm \Sigma_{0} (\bm L_{0}^T)^{-1})] + \sum_{j=1}^{p} \lambda_{j} \sum_{k<j}|l_{jk}| - \sum_{j=1}^{p} \lambda_{j} \sum_{k<j}|l_{0jk}|  \\
& \triangleq R_{1} + R_{2} + R_{3} + R_{4} + R_{5},
\end{align*}
where
\begin{align*}
R_{1} &= \log \left| \bm D \right| - \log \left| \bm D_{0} \right| + \tr[(\bm D^{-1} - \bm D_{0}^{-1})\bm D_{0}],\\
R_{2} &= \tr(\bm D^{-1} - \bm D_{0}^{-1})[\bm L^{-1} (\bm S - \bm \Sigma_{0}) (\bm L^T)^{-1}],\\
R_{3} &= \tr [\bm D_{0}^{-1} (\bm L^{-1} (\bm S - \bm \Sigma_{0}) (\bm L^T)^{-1} - \bm L^{-1}_{0} (\bm S - \bm \Sigma_{0}) (\bm L_{0}^T)^{-1})], \\
R_{4} &= \tr [\bm D^{-1} (\bm L^{-1} \bm \Sigma_{0} (\bm L^T)^{-1} - \bm L^{-1}_{0} \bm \Sigma_{0} (\bm L_{0}^T)^{-1})], \\
R_{5} &= \sum_{j=1}^{p} \lambda_{j} \sum_{k<j}|l_{jk}| - \sum_{j=1}^{p} \lambda_{j} \sum_{k<j}|l_{0jk}|.
\end{align*}
From the proof of Lemma 3 in \cite{kang2021variable}, we can derive

(1) $R_{1} \geq 1/8 \delta^4 \|\Delta_{D}\|_{F}^{2}$;

(2) for any $\epsilon > 0$, there exist constants $V_{1} > 0$ and $V_{2} > 0$ such that
$|R_{2}| \leq V_{1} \sqrt{p \log (p) / n} \|\Delta_{D}\|_{F}$;
and $R_{4} - |R_{3}| > 1 /2 \delta^4 \| \Delta_{L} \|_{F}^2 - V_{2} \sqrt{\log (p) / n} \sum_{(j, k) \in W^c} | l_{jk} | - V_{2} \sqrt{s \log (p) / n} \| \Delta_{L} \|_{F}$,
where the set $W = \{(j, k): k < j, l_{0jk} \neq 0 \}$.

Next, the term $R_{5}$ is further decomposed as
\begin{align*}
R_{5} = \sum_{j=1}^{p} \lambda_{j} \sum_{(j, k) \in W^c}|l_{jk}| + \sum_{j=1}^{p} \lambda_{j}
\sum_{(j, k) \in W} (|l_{jk}| - |l_{0jk}|) = R_{5}^{(1)} + R_{5}^{(2)},
\end{align*}
where $R_{5}^{(1)} = \sum_{j=1}^{p} \lambda_{j} \sum_{(j, k) \in W^c}|l_{jk}|$, and
\begin{align*}
| R_{5}^{(2)} | = | \sum_{j=1}^{p} \lambda_{j} \sum_{(j, k) \in W}(|l_{jk}| - |l_{0jk}|) |
\leq \sum_{j=1}^{p} \lambda_{j} \sum_{(j, k) \in W}| l_{jk} - l_{0jk} |
\leq \sum_{j=1}^{p} \lambda_{j} \sqrt{s} \| \Delta_{L} \|_{F}.
\end{align*}
Combine all the bounded terms $R_1$ to $R_5$ together, with probability greater than $1 - 2 \epsilon$, we have
\begin{align*}
&~~~ | G(\Delta_{L}, \Delta_{D}) | \\
&\geq R_{1} - |R_{2}| + R_{4} - |R_{3}| + R_{5}^{(1)} - |R_{5}^{(2)}|   \\
&\geq \frac{1}{8 \delta^4} \|\Delta_{D}\|_{F}^{2} - V_{1} \sqrt{p \log (p) / n} \|\Delta_{D}\|_{F} + \frac{1}{2 \delta^4} \| \Delta_{L} \|_{F}^2 - V_{2} \sqrt{\log (p) / n} \sum_{(j, k) \in W^c} | l_{jk} |   \\
&~~~ - V_{2}  \sqrt{s \log (p) / n} \| \Delta_{L} \|_{F} + \sum_{j=1}^{p} \lambda_{j} \sum_{(j, k) \in W^c}|l_{jk}| - \sum_{j=1}^{p} \lambda_{j} \sqrt{s} \| \Delta_{L} \|_{F}    \\
&= \frac{\kappa_{2} p \log(p)}{8 \delta^4 n} - \frac{V_{1} \sqrt{\kappa_{2}} p \log(p)}{n} + \frac{\kappa_{1} s \log(p)}{2 \delta^{4} n} - V_{2} \sqrt{\log (p) / n} \sum_{(j, k) \in W^c} | l_{jk} |  \\
&~~~ - \frac{V_{2} \sqrt{\kappa_{1}} s \log(p)}{n} + \sum_{j=1}^{p} \lambda_{j} \sum_{(j, k) \in W^c}|l_{jk}| - s \sqrt{\kappa_{1} \log (p) / n} \sum_{j=1}^{p} \lambda_{j} \\
&= \frac{\sqrt{\kappa_{2}} p \log(p)}{n} (\frac{\sqrt{\kappa_{2}}}{8\delta^4} - V_{1}) + \frac{\sqrt{\kappa_{1}} s p \log(p)}{n} (\frac{\sqrt{\kappa_{1}}}{2\delta^{4}} - \frac{\sum_{j=1}^{p} \lambda_{j}}{\sqrt{\log (p) / n}} - V_{2} )     \\
&~~~ + \sum_{(j, k) \in W^c}|l_{jk}| (\sum_{j=1}^{p} \lambda_{j} - V_{2} \sqrt{\log (p) / n}).
\end{align*}
Assume $\sum_{j=1}^{p} \lambda_{j} = V_{3}(\sqrt{\log (p) / n})$ where $V_{3} > V_{2}$, and choose $\kappa_{1} > 4 \delta^8 (K + V_{2})^2$ as well as $\kappa_{2} > 64\delta^8 V_{1}^2$, then $G(\Delta_{L}, \Delta_{D}) > 0$, which completes the proof.
\end{proof}

\begin{lemma}\label{Lemma:consistence1}
Suppose that the data are independently and identically distributed from $N(\bm 0, \mathbf{\Sigma}_0)$. Under the conditions C1 and C2, we have
\begin{align*}
\| \mathbf{\tilde{\Sigma}}_k - \mathbf{\Sigma}_0 \|_{F} = O_{p} \left( \sqrt{\frac{(s + p) \log (p)}{n}} \right),~~~ k = 1, ..., M.
\end{align*}
\end{lemma}

\begin{proof}
From the proof of Theorem 2 in \cite{kang2021variable}, it follows that
\begin{align*}
\| \tilde{\bm \Sigma}_k - \bm \Sigma_{0} \|_{F}^2 = O_{p}(\| \tilde{\bm L}_k - \bm L_{0} \|_{F}^2) +  O_{p}(\| \tilde{\bm D}_k - \bm D_{0} \|_{F}^2)
= O_{p}\left( \frac{(s + p)\log (p)}{n} \right),
\end{align*}
where the second equation results from Lemma \ref{lemma2_app}.
\end{proof}

Lemma \ref{Lemma:consistence1} establishes the convergence rate for each candidate estimate $\mathbf{\tilde{\Sigma}}_k$ for any variable orderings, based on which we prove the main results of Theorem \ref{theory1}.

\begin{proof}{\bf Proof of Theorem \ref{theory1}}\\
By Equation \eqref{eq:prop}, we obtain
\begin{align*}
\| \mathbf{\hat{\Sigma}}_{wae} - \mathbf{\Sigma}_0 \|_{F} = \| \sum_{k=1}^{M} \omega_k \mathbf{\tilde{\Sigma}}_k - \mathbf{\Sigma}_0 \|_{F} = \| \sum_{k=1}^{M} \omega_k ( \mathbf{\tilde{\Sigma}}_k - \mathbf{\Sigma}_0 ) \|_{F}.
\end{align*}
Based on the property of Frobenius norm, together with Lemma \ref{Lemma:consistence1}, we have
\begin{align*}
\| \mathbf{\hat{\Sigma}}_{wae} - \mathbf{\Sigma}_0 \|_{F} \leq \sum_{k=1}^{M} \| \omega_k (\mathbf{\tilde{\Sigma}}_k - \mathbf{\Sigma}_0) \|_{F} =
\sum_{k=1}^{M} \omega_k \| \mathbf{\tilde{\Sigma}}_k - \mathbf{\Sigma}_0 \|_{F}
= O_{p} \left( \sqrt{\frac{(s + p) \log (p)}{n}} \right).
\end{align*}
The proofs for the convergence rate of the estimator $\mathbf{\hat{\Sigma}}_{wae^\ast}$ follow the same principles.
\end{proof}

\end{document}